\documentclass[11pt, a4paper]{article}

\usepackage[utf8]{inputenc}     
\usepackage[T1]{fontenc}        
\usepackage{lmodern}            
\usepackage[english]{babel}     

\usepackage{authblk}            

\usepackage{geometry}
\usepackage{amsmath}
\usepackage{amssymb}
\usepackage{amsfonts}
\usepackage{amsthm}
\usepackage{mathtools}
\usepackage{amscd}
\usepackage[nobysame, alphabetic]{amsrefs}   

\usepackage{bm}
\usepackage{bbm}
\usepackage{mathrsfs}
\usepackage{accents}
\usepackage{slashed}
\usepackage{stmaryrd}
\usepackage{eucal}
\usepackage[all]{xy}
\usepackage{graphicx}
\usepackage[dvipsnames, svgnames]{xcolor}
\usepackage{tikz}
\usetikzlibrary{positioning, arrows.meta}
\usepackage{tikz-cd}

\definecolor{Strawberry}{HTML}{FB2943}
\definecolor{Raspberry}{HTML}{E30B5D}
\definecolor{DarkStrawberry}{HTML}{B01D2F}

\usepackage{hyperref}
\hypersetup{
    colorlinks=true,
    linkcolor=Strawberry,      
    citecolor=DarkStrawberry,  
    urlcolor=Raspberry,        
    pdftitle={Graph Quantum Magic Squares and Free Spectrahedra},
    pdfauthor={Francesca La Piana},
    pdfpagemode=UseOutlines,
    bookmarksnumbered=true,
    bookmarksopen=true
}

\theoremstyle{plain}
\newtheorem{theorem}{Theorem}[section]
\newtheorem{proposition}[theorem]{Proposition}
\newtheorem{definition}[theorem]{Definition}

\newtheorem{corollary}[theorem]{Corollary}

\theoremstyle{definition}

\theoremstyle{remark}
\newtheorem{remark}[theorem]{Remark}
\newtheorem{observation}[theorem]{Observation}
\newtheorem{question}[theorem]{Question}

\newcommand{\R}{\mathbb{R}}
\newcommand{\N}{\mathbb{N}}
\newcommand{\Z}{\mathbb{Z}}
\newcommand{\C}{\mathbb{C}}

  \newcommand{\cC}{\mathcal{C}}
  
  \newcommand{\cI}{\mathcal{I}}
  
\newcommand{\cM}{\mathcal{M}}  
\newcommand{\cP}{\mathcal{P}}

 \newcommand{\cZ}{\mathcal{Z}}

\newcommand{\cMn}{\mathcal{M}^{(n)}}
\newcommand{\cPn}{\mathcal{P}^{(n)}}
\newcommand{\cCn}{\mathcal{C}^{(n)}}
\newcommand{\cMG}{\mathcal{M}^{(\Gamma)}}
\newcommand{\cPG}{\mathcal{P}^{(\Gamma)}}
\newcommand{\cCG}{\mathcal{C}^{(\Gamma)}}

\usepackage{xparse}

\NewDocumentCommand{\Ms}{o}{
  \mathcal{M}_s%
  \IfValueT{#1}{^{(#1)}}
}
\NewDocumentCommand{\Ps}{o}{
  \mathcal{P}_s%
  \IfValueT{#1}{^{(#1)}}
}

\DeclareMathOperator{\Tr}{Tr}

\DeclareMathOperator{\Aut}{Aut}

\newcommand{\Her}[2][\C]{\operatorname{Her}_{#2}(#1)}
\newcommand{\Mat}[2]{\operatorname{Mat}_{#2}\!\left(#1\right)}
\newcommand{\Sym}[1]{\mathrm{Sym}_{#1}(\mathbb{R})}
\usepackage{xparse}

\NewDocumentCommand{\Comm}{o}{%
  \mathrm{Comm}%
  \IfValueT{#1}{_{\!#1}}%
}
\newcounter{remlist_counter}

\newcounter{proplist_counter}
\newlength{\maxlabelwidth}

\newcommand{\inst}[1]{\(^\textrm{#1}\) }
\usepackage[utf8]{inputenc}

\title{Graph Quantum Magic Squares and Free Spectrahedra}

\author{Francesca La Piana\inst{1}}
\date{
\parbox[t]{0.9\textwidth}{\footnotesize{%
\begin{itemize}
\item[1] Department of Mathematics, University of Oslo, P.O.\ Box 1053, 0316 Blindern, Oslo (Norway), e-mail: franla@math.uio.no
\end{itemize}
}}
\\
\vspace{\baselineskip}
\today}

\begin{document}

\maketitle

\begin{abstract}

Recently De les Coves, Drescher and Netzer showed that an analogue of the Birkhoff--von Neumann theorem fails in the quantum setting \cite{DDN20}. Motivated by this and 
questions arising in the study of quantum automorphisms of graphs, we
introduce a graph-based variant of quantum magic squares and show that the
analogue already fails for the cycle \(C_4\), via an explicit
counterexample. We also show that they admit monic linear matrix inequality descriptions, hence form compact
free spectrahedra.
\end{abstract}

{
\hypersetup{linkcolor=Raspberry}
\tableofcontents
}

\tableofcontents

\section{Introduction}

A \emph{magic square} is an \(n \times n\) matrix with nonnegative entries such that the sum of the elements in each row and in each column is the same; this common value is called the \emph{magic constant}. 
When the magic constant is equal to one, the matrix is typically referred to as a \emph{doubly stochastic matrix} (or bistochastic matrix).

The set of all \(n \times n\) doubly stochastic matrices forms a convex polytope known as the \emph{Birkhoff polytope}. 
Its vertices are exactly the permutation matrices, namely those matrices in which every row and every column contains exactly one entry equal to \(1\) and all other entries are \(0\).
This geometric structure is formalized by the classical \emph{Birkhoff--von Neumann theorem}, which states that every doubly stochastic matrix can be written as a convex combination of permutation matrices.

This classical setting generalizes to the non-commutative case, where the entries of the matrix are no longer numbers but positive semidefinite elements of an algebra or, more commonly, of a \(C^*\)-algebra.

\medskip

In their 2020 paper, De les Coves, Drescher, and Netzer studied \emph{quantum magic squares}, whose entries are complex square matrices, classifying them according to the properties of their entries.

The main question they raise is whether the classical Birkhoff--von Neumann theorem continues to hold in this noncommutative setting.
\begin{question}
Does the matrix convex hull of quantum permutation matrices, \(\cPn\), coincide with the full set of quantum magic squares, \(\cMn\), i.e.
\[
\operatorname{mconv}\bigl(\mathcal{P}^{(n)}\bigr) = \mathcal{M}^{(n)}?
\]
\end{question}
To address this, they adopt the framework of \emph{matrix convexity}, where the classical notion of extreme points is naturally replaced by that of \emph{Arveson extreme points}.  
Their main result establishes that the Birkhoff–von Neumann theorem fails in the quantum setting:

\begin{theorem}\label{thm:DDN}
For every \( n \ge 3 \),
\[
\mathcal{M}^{(n)} \neq \mathrm{mconv}(\mathcal{P}^{(n)}),
\]
that is, the matrix convex hull of quantum permutation matrices does not cover 
the full set of quantum magic squares.
\end{theorem}

They also observe that the set of quantum magic squares is a free spectrahedron, and thus
\(\mathcal{M}^{(n)}\) coincides with the matrix convex hull of its Arveson extreme points. 
Moreover, every quantum permutation matrix is an Arveson extreme point.
However, in view of their main theorem, it follows that not every Arveson extreme point is a quantum permutation matrix.

\medskip

Their work provides the main foundation for this paper, which we view as a continuation of their approach. Motivated by these results, we introduce a  variant of quantum magic squares subject to constraints imposed by a graph \(\Gamma\), where the underlying combinatorial structure is encoded by the adjacency matrix of \(\Gamma\).
We call these objects \emph{graph quantum magic squares} (GQMS).

In analogy with \cite{DDN20}, we may ask for the graph-analogue 
of the quantum Birkhoff--von Neumann question. Recall that 
\[
\cM^{(\Gamma)} := \{\,A \in \cM^{(n)} \mid 
A (I_s \otimes A_\Gamma) = (I_s \otimes A_\Gamma) A \,\}
\]
is the set of quantum magic squares that commute with the adjacency matrix 
\(A_\Gamma\) of the graph \(\Gamma\) (for precise definitions see Definition~\ref{def:GQMS}). 

\begin{question}
Given a graph \(\Gamma\), does the matrix convex hull of its quantum permutation 
matrices \(\cPG\) generate all graph quantum magic squares? That is,
\[
\mathrm{mconv}(\cP^{(\Gamma)}) = \cM^{(\Gamma)} \; ?
\]
\end{question}

This notion allows us to extend the framework of \cite{DDN20} to a setting where the magic relations are combined with commutation constraints imposed by the graph.

\medskip

The paper is organized as follows.
In Section~2, we introduce the notation and some background.
We briefly recall the definition of classical magic squares and then review the setting used by De~las~Cuevas, Drescher, and Netzer to define quantum magic squares, together with their classification according to the properties of the entries.
We also summarize the main tools employed in their 2020 work to prove the failure of the Birkhoff--von Neumann theorem in the quantum case.
Finally, we recall the definition of free spectrahedra following Evert and Helton’s framework~\cite{EH19}, reformulating it in terms of linear matrix inequalities (LMI).

In Section~3, we introduce the concept of \emph{graph quantum magic squares}, 
which can be viewed as a variant of quantum magic squares where additional constraints 
are imposed by requiring the matrix to commute with the adjacency matrix of a given graph.

Following the same structure as in~\cite{DDN20}, we classify graph quantum 
magic squares according to the properties of their entries and provide a counterexample 
to the Birkhoff--von Neumann theorem in the case of the cycle graph~\(C_4\), 
by constructing a matrix that belongs to~\(\mathcal{M}^{(C_4)}\) but not to 
\(\mathrm{mconv}(\mathcal{P}^{(C_4)})\).

We then show explicitly that the set of quantum magic squares defined in~\cite{DDN20} forms a free spectrahedron, and that the same holds for graph quantum magic squares, at least for \(k\)-regular graphs.

The paper concludes with an appendix, where we provide the Hermitian basis used in the numerical construction of the counterexample, and explicitly compute the dimension of the commutant for the family of cycle graphs \(C_n\).

\section{Notations and background}
In this section, we introduce some notation and basic definitions that will be used throughout the paper.
We also briefly recall a few standard concepts that will be useful in the following sections.

\begin{itemize}
    \item \(\operatorname{Mat}_n(S)\) denotes the space of \(n \times n\) matrices with entries in the set \(S\).
    \item \(\Her[\C]{n} = \{ A \in \operatorname{Mat}_n(\C) \mid A^*=A \}\) is the real vector space of Hermitian \(n \times n\) complex matrices.
    \item \(A \succeq 0\) means that the Hermitian matrix \(A\) is positive semidefinite.
    \item \(\operatorname{Psd}_n(\C)\) denotes the convex cone of all positive semidefinite \(n \times n\) matrices over \(\C\).
    \item \(\Sym d\) the set real symmetric \(d \times d\) matrices.
\end{itemize}

\subsection{Quantum Magic Squares} In the noncommutative setting, the entries of a magic square are elements of a unital \(C^*\)-algebra. A natural generalization of classical magic squares is to consider block matrices whose entries are positive elements such that the sum on each row and column equals the unit of the \(C^*\)-algebra. Throughout, we restrict to the case of \(\operatorname{Mat}_s(\C)\). 

To describe the “row sums” and “column sums’’ in this context, it is convenient to recall 
the notion of operator-valued measure.

A \emph{positive operator-valued measure} (POVM) on \(\C^s\) is a finite family of positive 
semidefinite matrices \(A_i \in \operatorname{Psd}_s(\C)\) such that 
\[
\sum_i A_i = I_s.
\]
If, in addition, each \(A_i\) is a projection (\(A_i^2 = A_i = A_i^*\)), the family is called a 
\emph{projection-valued measure} (PVM). Thus, PVMs form a subset of POVMs.

We now define the main object of our study.

\begin{definition} 
A \emph{quantum magic square} of external size \(n\) and internal size \(s\) is a block matrix 
\[
A=\begin{pmatrix}
A_{11} & \cdots & A_{1n}\\
\vdots & \ddots & \vdots\\
A_{n1} & \cdots & A_{nn}
\end{pmatrix}
\qquad \text{with } A_{ij}\in \operatorname{Psd}_s(\C),
\]
such that, for every \(i,j\),
\begin{equation}\label{eq:magic}
    \sum_{k=1}^n A_{ik} = I_s 
    \qquad \text{and} \qquad 
    \sum_{k=1}^n A_{kj} = I_s.
\end{equation}
We refer to~\eqref{eq:magic} as the \emph{magic relations}.
Equivalently, we say that each row and each column forms a POVM.
\end{definition}
Following the convention of \cite{DNV23}, we focus on the following three families of quantum magic square according to the properties of their entries. More precisely, for any \(n,s \in \N\), we define respectively the set of \emph{quantum magic square}, the set of \emph{quantum permutation matrices} and the set \emph{commuting quantum permutation matrices} in the following way:
\[ 
\mathcal{M}_s^{(n)} := \Bigl\{ A \in \operatorname{Mat}_n\bigl(\operatorname{Psd}_s(\C)\bigr) \mid \sum_{k=1}^n A_{ik} = I_s,\ \sum_{k=1}^n A_{kj} = I_s \ \forall i,j \Bigr\}, \] \[ \mathcal{P}_s^{(n)} := \Bigl\{ A \in \mathcal{M}_s^{(n)} \mid A_{ij} = A_{ij}^2 = A_{ij}^* \ \forall i,j \Bigr\},
\]
\[ 
\mathcal{C}_s^{(n)} := \Bigl\{ A \in \mathcal{P}_s^{(n)} \mid [A_{ij},A_{kl}]=0 \ \forall i,j,k,l \Bigr\}.
\]

Taking the union over all internal size \(s\), we have the total sets \(\mathcal{M}^{(n)} := \bigcup_{s\in\N} \mathcal{M}_s^{(n)}\), \(\mathcal{P}^{(n)}:= \bigcup_{s\in\N} \mathcal{P}_s^{(n)}\), \(\mathcal{C}^{(n)}:= \bigcup_{s\in\N} \mathcal{C}_s^{(n)}\), and they satisfy
\[ 
\mathcal{C}^{(n)} \subseteq \mathcal{P}^{(n)} \subseteq \mathcal{M}^{(n)} \qquad (\forall n\in\N).
\] 

\begin{observation} For \(n=1,2,3\) we have \(\mathcal{C}^{(n)} = \mathcal{P}^{(n)}\). In particular, every \(3\times 3\) quantum permutation matrix is also commuting (see, e.g., \cite{Web23}*{Lemma~2.5}). \end{observation} 
As in \cite{DDN20} and \cite{DNV23} we use the
framework of \emph{matrix convexity}, a dimension-free notion of
convexity, which generalizes classical convexity. This framework plays a central role in the study of noncommutative analogues of the
Birkhoff--von Neumann theorem.

\begin{definition}[Matrix convexity] 
Let \(\{R_s\}_{s\in\N}\) with \(R_s \subseteq \operatorname{Mat}_n\bigl(\Her{s}\bigr)\), and set \(R:=\bigcup_{s\in\N} R_s\). We say that \(R\) is matrix convex if, for any \(s,t\in\N\), any \(A^{(1)},\dots,A^{(k)}\in R_{s_i}\), and any \(V_1,\dots,V_k\in \operatorname{Mat}_{s_i,t}(\C)\) with \[ \sum_{i=1}^k V_i^* V_i = I_t, \] one has \[ \sum_{i=1}^k V_i^*\, A^{(i)} \, V_i \ \in\ R_t. \] Here the conjugation acts blockwise, i.e., \(\bigl(\sum_i V_i^* A^{(i)} V_i\bigr)_{kl} = \sum_i V_i^* A^{(i)}_{kl} V_i\). Equivalently, we say that \(R\) is closed under matrix-valued convex combinations. 
\end{definition} 

For a subset \(S\subseteq \operatorname{Mat}_n(\Her{s})\), the matrix convex hull of \(S\), denoted as \(\operatorname{mconv}(S)\), is the smallest matrix convex set containing \(S\). 
\begin{remark}
If the  \(V_i\) are scalar matrices, i.e. \(V_i = \lambda_i I_s\) with \(\lambda_i\ge0\) and \(\sum_i \lambda_i = 1\), then the definition recovers the usual convex combinations:
\[
A = \sum_{i=1}^k \lambda_i A^{(i)}.
\]
\end{remark}
The main result of \cite{DDN20} shows that the classical
Birkhoff--von Neumann theorem fails already in the smallest
non-commutative setting.

The key tool used for proving this failure is the following criterion:
\begin{proposition}[{\cite{DDN20}*{Proposition~18}}]\label{prop:DDN-separation}
Let \(A\in \cM^{(n)}_s\) and let
\[\begin{split}
    \mathrm{col}(A) &:= \sum_{i,j=1}^n e_i\otimes e_j \otimes A_{ij} \in \C^n\otimes\C^n\otimes \Her[\C]{s}, \\
\mathrm{diag}(A) &:= \sum_{i,j=1}^n E_{ii}\otimes E_{jj}\otimes A_{ij}\in \Mat{\C}{n}\otimes \Mat{\C}{n}\otimes \Her{s}.
\end{split}\]
Define
\[
\varphi(A) := \mathrm{diag}(A) - \mathrm{col}(A)\,\mathrm{col}(A)^*
\]
and
\[
\psi(A) := \sum_{\substack{i\neq j \\ k\neq \ell}}
E_{ij}\otimes E_{k\ell} \otimes 
\Bigl(-\alpha_n\,I_s + \beta_n A_{ik} + \beta_n A_{j\ell} + \gamma_n A_{i\ell} +\gamma_n A_{jk}\Bigr),
\]
where
\[
\alpha_n = \frac{1}{(n-1)(n-2)},\qquad
\beta_n = \frac{n-1}{n(n-2)},\qquad
\gamma_n = \frac{1}{n(n-2)}.
\]

Let \(\cZ_e^{(n)} := \{Z\in \operatorname{Mat}_n(\C)\mid \mathrm{diag}(Z)=0,\; Z\,\mathbf{1}=0\}\) and set
\[
\mathcal{S}^{(n)} := \bigl(\cZ_e^{(n)}\otimes \cZ_e^{(n)} \otimes \operatorname{Mat}_s(\C)\bigr)_{\mathrm{her}}.
\]

If \(A \in \mathrm{mconv}(\cP^{(n)})_s\), then there exists 
\(X \in \mathcal{S}^{(n)}\) such that
\[
\varphi(A) + \psi(A) + X \succeq 0.
\]
\end{proposition}
We recall the main result of \cite{DDN20}.

\begin{theorem}[{\cite{DDN20}*{Theorem~16}}]
For every \(n \ge 3\) we have
\[
\operatorname{mconv}(\cP^{(n)}) \subsetneq \cM^{(n)}.
\]
In particular, the equality \(\cP^{(3)} = \cC^{(3)}\) holds, and still
\[
\operatorname{mconv}(\cP^{(3)}) \subsetneq \cM^{(3)}.
\]
\end{theorem}

 Thus, the matrix convex hull of quantum permutation matrices does not recover the full set of quantum magic squares. This motivates the study of \emph{graph quantum magic squares}, which impose additional constraints imposed by the structure of a graph. In Section~\ref{sec:free-spec} we show that this variant also admits a free-spectrahedral description.

\subsection{Free spectrahedra and linear matrix inequalities}

We now recall the notion of \emph{free spectrahedra}, which provides a description of matrix convex sets 
through \emph{linear matrix inequalities} (LMIs).   
In this section we follow the conventions of \cite{EH19}, adapted to our notation.

\medskip

A free spectrahedron is a matrix convex set that can be described by a linear matrix inequality.  

Fixing a tuple \(A = (A_1,\dots,A_g) \in ( \Sym s )^g\), we define the \emph{monic linear pencil} associated to \(A\) as 
\[
L_A(x) \;=\; I_d \;+\; A_1 x_1 \;+\; A_2 x_2 \;+\; \cdots \;+\; A_g x_g,
\]
where \(x = (x_1,\dots,x_g)\) is a tuple of noncommuting variables.  

For a tuple \(X = (X_1,\dots,X_g) \in ( \Her s )^g\) of real symmetric \(s \times s\) matrices, 
the evaluation of \(L_A\) at \(X\) is defined by
\[
L_A(X) \;=\; I_d \otimes I_s \;+\; A_1 \otimes X_1 \;+\; A_2 \otimes X_2 \;+\; \cdots \;+\; A_g \otimes X_g,
\]
where \(\otimes\) denotes the Kronecker product.  

A \emph{linear matrix inequality  } is an inequality of the form
\[
L_A(X) \;\succeq\; 0.
\]

Denoting the homogeneous linear part of the pencil by
\[
\Lambda_A(X) \;:=\; A_1 \otimes X_1 \;+\; A_2 \otimes X_2 \;+\; \cdots \;+\; A_g \otimes X_g,
\]
we can write the evaluation of \(L_A\) at \(X\) as
\[
L_A(X) \;=\; I_{ds} + \Lambda_A(X).
\]
\begin{definition}[Free spectrahedron at level \(s\)]
Given a tuple \(A \in (\Her d)^g\) and a positive integer \(s\), the \emph{free spectrahedron at level \(s\)} is
\[
\mathcal{D}_A(s) := \{\, X \in (\Her s)^g \mid L_A(X) \succeq 0 \,\}.
\]
That is, \(\mathcal{D}_A(s)\) consists of all \(g\)-tuples of \(s \times s\) Hermitian matrices \(X\)
such that the evaluation \(L_A(X)\) is positive semidefinite.
\end{definition}

\begin{definition}[Free spectrahedron]
Given a tuple \(A \in (\Her d)^g\) and a positive integer \(s\), the \emph{free spectrahedron} associated with \(A\) is defined as the graded set
\[
\mathcal{D}_A := \bigcup_{s \ge 1} \mathcal{D}_A(s).
\]
\end{definition}

We now recall the notions of dilations and Arveson extreme points,
following \cite{DDN20} and \cite{EH19}.

Given \(X \in \mathcal{D}_A(s)\), a block matrix
\[
Y =
\begin{pmatrix}
X & \beta \\
\beta^* & \gamma
\end{pmatrix}
\in \mathcal{D}_A(s+\ell),
\qquad \ell \ge 1,
\]
is called a \emph{dilation} of \(X\).
The dilation is \emph{trivial} if \(\beta = 0\), in which case \(Y\) is a direct sum of \(X\)
with another element of \(\mathcal{D}_A\).

\begin{definition}[Arveson extreme point]\label{def:arveson}
Let \(\mathcal D = \bigcup_{s \ge 1} \mathcal D(s)\) be a matrix convex set.
An element \(X \in \mathcal D(s)\) is an \emph{Arveson extreme point} of \(\mathcal D\)
if it admits no nontrivial dilation inside \(\mathcal D\).
That is, whenever
\[
Y \in \mathcal D(s+\ell)
\quad \text{and} \quad
Y =
\begin{pmatrix}
X & \beta \\
\beta^* & \gamma
\end{pmatrix},
\]
for some \(\ell \ge 1\), one must have \(\beta = 0\).
\end{definition}
\begin{remark}
In the literature (see e.g.~\cite{EH19}), an Arveson extreme point is usually
required to be \emph{irreducible}, meaning that the underlying tuple admits no
nontrivial common reducing subspace.
Without this, Definition~\ref{def:arveson} characterizes
direct sums of Arveson extreme points.
Since all extremal objects considered in this paper are irreducible, this
distinction will not play a role in what follows.
\end{remark}
\begin{remark}[Arveson extreme points and spanning]\label{rem:arveson-span}
Let \(\mathcal D\) be a compact free spectrahedron which is closed under complex conjugation.
In particular, this hypothesis holds whenever the defining pencil has real coefficient matrices (since then \(X \in \mathcal D(s)\) implies \(\overline{X}\in \mathcal D(s)\)).
This will be the case for the spectrahedra arising from quantum magic squares and from the graph ones considered later, since the corresponding coefficient matrices can be chosen real. By \cite{EH19}*{Theorem~1.3}, such a set satisfies
\[
\mathcal D
=
\operatorname{mconv}\bigl(\operatorname{AbsExt}(\mathcal D)\bigr),
\]
where \(\operatorname{AbsExt}(\mathcal D)\) denotes the set of absolute extreme points.

Since every absolute extreme point is, by definition, an Arveson extreme point,
we have
\[
\operatorname{AbsExt}(\mathcal D)
\subseteq
\operatorname{ArvesonExt}(\mathcal D).
\]
As \(\mathcal D\) is matrix convex, this implies
\[
\mathcal D
=
\operatorname{mconv}\bigl(\operatorname{ArvesonExt}(\mathcal D)\bigr).
\]
\end{remark}

\begin{theorem}[\cite{EH19}]\label{thm:EvertHelton}
Let \(\mathcal{D}\) be a compact free spectrahedron which is closed under complex conjugation.
Then
\[
\mathcal{D}
=
\operatorname{mconv}\bigl(\operatorname{ArvesonExt}(\mathcal{D})\bigr).
\]
\end{theorem}
\subsection{Quantum Automorphism of Graphs}\label{sec:QAutGamma}

To motivate the definition of Graph Quantum Magic Squares, we recall the notion of the quantum automorphism group of a finite graph, (see e.g. \cite{Ban05}), following Wang's definition of quantum permutation groups \cite{Wan95}.

In the classical setting, an automorphism of a graph \(\Gamma\)
with adjacency matrix \(A_\Gamma\)
is a permutation matrix \(P\in \mathrm{Mat}_n(\{0,1\})\)
such that \(PA_\Gamma=A_\Gamma P\). The set of all such matrices forms the automorphism  \(\mathrm{Aut(\Gamma)}\).

In the quantum setting, instead of scalars
\(\{0,1\}\), the entries
are elements of a unital
C\(^*\)-algebra. A matrix
is called a \emph{quantum permutation matrix} if its entries are projections \(p_{ij}\) such that every row and column sums to the identity.

The quantum automorphism group of \(\Gamma\), denoted by
\(\mathrm{Aut}^+(\Gamma)\), is defined via its universal
C\(^*\)-algebra \(C(\mathrm{Aut^+}(\Gamma))\). This algebra is generated by \(n^2\) elements \(u_{ij}\) subject to the relations making a quantum permutation matrix that commutes with the adjacency matrix:
    \[ U A_\Gamma = A_\Gamma U. \] 
    
For some graphs, this algebra is non-commutative, implying the existence of genuine quantum symmetries. A graph \(\Gamma\) is said to have \emph{no quantum symmetry} if the algebra
is commutative, implying \(C(\mathrm{Aut}^+(\Gamma))=C(\mathrm{Aut}(\Gamma))\).

On the contrary, if the algebra is non-commutative, the graph has genuine quantum symmetries.
This distinction is fundamental to our work. Our definition of graph quantum magic squares can be seen as a relaxation of \(C(\mathrm{Aut}^+(\Gamma))\): instead of requiring that the entries be projections (as in \(C(\mathrm{Aut}^+(\Gamma))\)), we only require that they are positive semi-definite operators (POVM), retaining the commutation condition.

In the next section, we introduce the notion of \emph{Graph Quantum Magic Squares} and provide explicit
monic linear pencils that describe both quantum magic squares
and the graph ones.

\section{Graph Quantum Magic Squares}

In this section, we introduce a graph-variant of quantum magic squares, 
obtained by imposing commutation with the adjacency matrix of a graph~\(\Gamma\).
The motivation is to explore how the combinatorial symmetries of a graph
interact with the algebraic structure of quantum magic squares.

\subsection{Definitions and preliminary observations}

Let \(\Gamma = (V,E)\) be a simple graph with \(|V| = n\), and let \(A_\Gamma\) denote
its adjacency matrix. Here, the external dimension \(n\) of the matrix corresponds to 
the number of vertices of \(\Gamma\).

\begin{definition}[Graph quantum magic square]
Let \(\Gamma\) be a graph on \(n\) vertices and \(A_\Gamma\) its adjacency matrix.
A quantum magic square \(X = (X_{ij})_{1\le i,j\le n} \in \Mat{\Her s}{n}\)
is called a \emph{graph quantum magic square} if
\[
X\,(I_s \otimes A_\Gamma)
=
(I_s \otimes A_\Gamma)\,X,
\]
i.e., if it commutes with the adjacency matrix of the graph tensored with the identity matrix of the internal dimension \(s\).

\end{definition}
The tensor \(I_s \otimes A_\Gamma\) imposes the structure of the graph on the internal entries of the matrix \(s\times s\), leaving the external indices unchanged.

This commutation condition can be written in terms of internal blocks as
\[
\sum_{k \sim j} X_{ik}
=
\sum_{k \sim i} X_{kj}
\qquad \forall\, i,j \in \{1,\dots,n\},
\]
where \(k\sim j\) means that vertices \(k\) and \(j\) are adjacent in \(\Gamma\).

Following the classification in \cite{DDN20}, we define the following sets:
\begin{definition}[Graph Quantum Magic Square]\label{def:GQMS}
Let \(\Gamma\) be a finite graph with adjacency matrix \(A_\Gamma\), and let \(s\ge1\).
We define the following sets of graph quantum magic squares:
    
\[
\begin{aligned}
\cMG_s
&:= \{\,X \in \cMn_s \mid X (I_s \otimes A_\Gamma) = (I_s \otimes A_\Gamma) X \,\},\\
\cPG_s
&:= \{\,X \in \cMG_s \mid X_{ij}^2 = X_{ij} = X_{ij}^* \text{ for all } i,j \,\},\\
\cCG_s
&:= \{\,X \in \cPG_s \mid X_{ij}X_{kl} = X_{kl}X_{ij} \text{ for all } i,j,k,l \,\}.
\end{aligned}
\]
\end{definition}
As usual we set
\[
\cMG = \bigcup_{s\in\N} \cMG_s,
\qquad
\cPG = \bigcup_{s\in\N} \cPG_s,
\qquad
\cCG = \bigcup_{s\in\N} \cCG_s,
\]
and we always have
\[
\cCG \subseteq \cPG \subseteq \cMG.
\]

Following the main question in \cite{DDN20}, we consider its graph analogue.  
In this setting, we consider the matrix convex hull of the graph quantum permutation matrices, denoted by \(
\mathrm{mconv}(\cPG)\), and we ask whether it covers the set of all graph quantum magic squares:
\[
\mathrm{mconv}(\cPG)\overset{?}{=}\cMG.
\]

As already observed in the non-graph case \cite{DDN20}, working directly with \(\cPn\) is often difficult, whereas the set \(\cCn\) is usually much easier to handle.  
For this reason, in the graph setting it is useful to understand when 
\(\cPG\) and \(\cCG\) coincide, so that we may replace 
\(\cPG\) with the set \(\cCG\).

Understanding whether \(\cPG = \cCG\) is equivalent to asking whether the graph has no quantum symmetry, i.e., whether its quantum automorphism group reduces 
to the classical one. As remarked in \autoref{sec:QAutGamma} that a graph \(\Gamma\) has no quantum symmetry if
its quantum automorphism group \(\Aut^{+}(\Gamma)=\Aut(\Gamma)\).

Thanks to the classification results in \cite{BB07} and \cite{Sch18},
many families of graphs with this property are known. For instance, the complete graphs \(K_2\) 
and \(K_3\), the cycle graphs \(C_n\) for \(n \neq 4\), and the Petersen graph all 
have no quantum symmetries, and therefore satisfy
\[
\cCG = \cPG.
\]

\subsection{A counterexample for the square graph \(C_4\)}\label{sec:counterexample}

In the general (non-graph) quantum magic square setting,
\cite{DDN20} showed that the Birkhoff--von Neumann
theorem does not hold in the quantum case.

The following result is the graph analogue of Theorem~\ref{thm:DDN}.

\begin{theorem}\label{thm:C4-counterexample}
There exists a \(C_4\)-quantum magic square 
\(B \in \cM^{(C_4)}_2\)
such that
\[
B \notin \mathrm{mconv}\!\bigl(\cP^{(C_4)}\bigr)_2.
\]
In particular,
\[
\mathrm{mconv}\!\bigl(\cP^{(C_4)}\bigr)
\subsetneq 
\cM^{(C_4)},
\]
so the Birkhoff--von Neumann theorem fails for \(C_4\) already
at internal dimension \(s=2\).
\end{theorem}

\begin{proof}
We start from the same counterexample matrix
\(A \in \cM_2^{(4)}\) used in \cite{DDN20}.
The matrix \(A\) is a quantum magic square but does not commute with
the adjacency matrix of \(C_4\), hence it is not suitable for the
graph-constrained setting.

To impose the \(C_4\)-symmetry, we average \(A\) along the cyclic
group of graph automorphisms:
\[
  B_{ij}
  := \frac{1}{4} \sum_{k=0}^3 A_{\,i+k,\, j+k},
  \qquad (\text{indices mod } 4).
\]
The averaging preserves positivity and the magic relations, so
\(B\) is still a QMS.
The resulting
matrix satisfies
\[
A_{C_4} B = B A_{C_4},
\]
and therefore \(B \in \cM^{(C_4)}\).

Proposition~\ref{prop:Z4-averaging} gives more details about the averaging construction.

Now, we recall the necessary condition for a matrix \(A\) to be in the matrix convex hull of the set of quantum permutation matrices
(Proposition~\ref{prop:DDN-separation}).
If \(A \in \operatorname{mconv}(\mathcal{P}^{(4)})_2\), then there exists
\[
X \in \bigl(\cZ_e^{(4)} \otimes \cZ_e^{(4)}
           \otimes \Mat{\C}{2}\bigr)_{\mathrm{her}}
\]
such that
\[
\varphi(A) + \psi(A) + X \succeq 0.
\]
We apply this to \(A=B\) and set
\[
M_B := \varphi(B) + \psi(B).
\]
Thus, if \(B\) belonged to \(\mathrm{mconv}(\cP^{(4)})_2\) (and hence to
\(\mathrm{mconv}(\cP^{(C_4)})_2\)), there would exist
\[
X \in \bigl(\cZ_e^{(4)} \otimes \cZ_e^{(4)}
           \otimes \Mat{\C}{2}\bigr)_{\mathrm{her}}
\]
such that
\begin{equation}\label{eq:MB-plus-X-psd}
M_B + X \succeq 0.
\end{equation}

Let \(\{z_a\}\) be a basis of \(\cZ_e^{(4)}\) and
\(\{S_j\}\) a Hermitian basis of \(\Mat{\C}{2}\).
Every admissible \(X\) can then be written as
\[
X = \sum_{a,b,j} \xi_{a,b,j}\,
     Z_{a,b,j},
\qquad
Z_{a,b,j} := z_a \otimes z_b \otimes S_j.
\]

To exclude the possibility of an X satisfying \eqref{eq:MB-plus-X-psd}, we use the dual formulation.

Consider a positive semidefinite matrix \(Y \succeq 0\) such that
\begin{equation}\label{eq:dual-conditions}
\Tr\!\bigl(Y Z_{a,b,j}\bigr)=0
\quad\text{for all }a,b,j,
\qquad\text{and}\qquad
\Tr\!\bigl(Y M_B\bigr) < 0.
\end{equation}
(The construction of such a matrix \(Y\), via a semidefinite program,
is described in Appendix~\ref{sec:dual-certificate}.)  
Then, for any admissible \(X\) as above, we have
\[
\Tr\!\bigl(YX\bigr)
 = \sum_{a,b,j} \xi_{a,b,j}\,
   \Tr\!\bigl(Y Z_{a,b,j}\bigr)
 = 0
\]
by the first condition in~\eqref{eq:dual-conditions}, and hence
\[
\Tr\!\bigl(Y(M_B+X)\bigr)
 = \Tr\!\bigl(Y M_B\bigr) + \Tr\!\bigl(YX\bigr)
 = \Tr\!\bigl(Y M_B\bigr)
 < 0.
\]
In particular, \(M_B+X\) cannot be positive semidefinite, so
\eqref{eq:MB-plus-X-psd} fails for every admissible \(X\).
By Proposition~\ref{prop:DDN-separation}, this implies
\[
B \notin \operatorname{mconv}(\mathcal{P}^{(4)})_2,
\]
and therefore
\(B \notin \operatorname{mconv}(\cP^{(C_4)})\).
\end{proof}

Next proposition justify the averaging used in Theorem~\ref{thm:C4-counterexample}.

\begin{proposition}\label{prop:Z4-averaging}
Let \(A=\bigl(A_{ij}\bigr)_{i,j=1}^4 \in \operatorname{Mat}_4(\Her{s})\) be a \(4\times 4\) quantum magic square of internal size \(s\).
Let \(P_{C_4}\) be the \(4\times 4\) permutation matrix
\[
P_{C_4} e_i = e_{\,i+1 \pmod 4}, \qquad i=1,\dots,4,
\]
that is,
\[
P_{C_4} =
\begin{pmatrix}
0 & 0 & 0 & 1\\
1 & 0 & 0 & 0\\
0 & 1 & 0 & 0\\
0 & 0 & 1 & 0
\end{pmatrix}.
\]
Define
\[
\Phi(A) := \frac{1}{4}\sum_{k=0}^{3} 
\bigl(P_{C_4}^{k}\otimes I_s\bigr)\,
A\,
\bigl(P_{C_4}^{k}\otimes I_s\bigr)^{*}.
\]
Then:
\begin{enumerate}
    \item[(i)] \(\Phi(A)\) is a quantum magic square;
    \item[(ii)] \(\Phi(A)\) commutes with the adjacency matrix of \(C_4\), i.e.
    \[
    (A_{C_4}\otimes I_s)\, \Phi(A) = \Phi(A)\, (A_{C_4}\otimes I_s);
    \]
    \item[(iii)]  we have (entrywise):
    \[
    \Phi(A)_{ij}
      = \frac{1}{4}\sum_{k=0}^3 A_{\,i+k,\; j+k},
    \qquad (\text{indices } i,j \text{ mod } 4).
    \]
\end{enumerate}
\end{proposition}

\begin{proof}
Let us set \(\widetilde P_{C_4} := P_{C_4}\otimes I_s\).  
Then \(\widetilde P_{C_4}\) is unitary, so \(\widetilde P^*_{C_4}=\widetilde P^{-1}_{C_4}\) and,
for every \(k \in \Z\),
\[
\bigl(\widetilde P_{C_4}^{\,k}\bigr)^*
= \bigl(\widetilde P_{C_4}^*\bigr)^{k}
= \widetilde P_{C_4}^{-k}.
\]
Moreover \((P_{C_4})^4 = I_4\) implies \(\widetilde P_{C_4}^{\,4} = I_4\otimes I_s\).

\smallskip

(ii) 
We start showing that \(\Phi(A)\) commutes with \(\widetilde P_{C_4}\).
By definition of \(\Phi\) and linearity we have
\begin{align*}
\widetilde P_{C_4}\,\Phi(A)\,\widetilde P_{C_4}^{-1}
&= \widetilde P_{C_4} \left( \frac{1}{4}\sum_{k=0}^{3} 
\widetilde P_{C_4}^{\,k} A \bigl(\widetilde P_{C_4}^{\,k}\bigr)^{*} \right) \widetilde P_{C_4}^{-1} \\
&= \frac14 \sum_{k=0}^3 
\widetilde P_{C_4}^{\,k+1} A \widetilde P_{C_4}^{-(k+1)}.
\end{align*}
Changing the index \(m=k+1 \pmod 4\), the sum becomes
\[
\frac14 \sum_{m=0}^3 \widetilde P_{C_4}^{\,m} A \widetilde P_{C_4}^{-m} = \Phi(A),
\]
so \(\widetilde P_{C_4}\,\Phi(A)\,\widetilde P_{C_4}^{-1}=\Phi(A)\), i.e. 
\begin{equation}\label{comm_P_tilde}
    \widetilde P_{C_4}\,\Phi(A)=\Phi(A)\,\widetilde P_{C_4}.
\end{equation}

The adjacency matrix of the undirected cycle is
\[
A_{C_4} = P_{C_4} + P_{C_4}^*,
\qquad
A_{C_4}\otimes I_s
= \widetilde P_{C_4} + \widetilde P_{C_4}^{*}.
\]
Since \(\Phi(A)\) is a convex combination of conjugates of \(A\), if \(A\) is Hermitian \(\Phi(A)\) is too. Taking adjoints in \eqref{comm_P_tilde}, we get
\[
\Phi(A)\,\widetilde P_{C_4}^{*} = \widetilde P_{C_4}^{*} \,\Phi(A).
\]
So \(\Phi(A)\) commutes with both \(\widetilde P_{C_4}\) and \(\widetilde P_{C_4}^{*}\), and so does their sum:
\[
(A_{C_4}\otimes I_s)\,\Phi(A)
= \Phi(A)\,(A_{C_4}\otimes I_s).
\]
This proves (ii).

\smallskip

(iii)
For a block matrix \(M\), the \((i,j)\)-th block is
\(M_{ij}=(e_i^*\otimes I_s)\,M\,(e_j\otimes I_s)\).

For a fixed \(k\), we have
\begin{align*}
\bigl(\widetilde P_{C_4}^{\,k} A \widetilde P_{C_4}^{-k}\bigr)_{ij}
&= (e_i^*\otimes I_s)\, \widetilde P_{C_4}^{\,k} A \widetilde P_{C_4}^{-k}\, (e_j\otimes I_s) \\
&= \bigl((e_i^* P_{C_4}^{\,k})\otimes I_s\bigr)\, A \,
   \bigl((P_{C_4}^{-k} e_j)\otimes I_s\bigr).
\end{align*}
Using the equalities \(P_{C_4}^{-k}e_j=e_{\,j-k}\) and
\(e_i^* P_{C_4}^{\,k} = (P_{C_4}^{-k}e_i)^* = e_{\,i-k}^*\), we get
\[
\bigl(\widetilde P_{C_4}^{\,k} A \widetilde P_{C_4}^{-k}\bigr)_{ij}
= (e_{\,i-k}^* \otimes I_s)\, A\, (e_{\,j-k}\otimes I_s)
= A_{\,i-k,\; j-k}.
\]
Therefore,
\[
\Phi(A)_{ij}
= \left(\frac14 \sum_{k=0}^3 \widetilde P_{C_4}^{\,k} A \widetilde P_{C_4}^{-k}\right)_{ij}
= \frac14 \sum_{k=0}^3 A_{\,i-k,\; j-k}.
\]
Since the indices are taken modulo \(4\), summing over \(k\) is equivalent to sum over \(-k\), and then we can write
\[
\Phi(A)_{ij}
= \frac14 \sum_{k=0}^3 A_{\,i+k,\; j+k}
\]
proving (iii).

\smallskip

(i)
Recall that a quantum magic square \(A\) satisfies:
\begin{enumerate}
\item[(a)] \(A_{ij} \succeq 0\) for all \(i,j\);
\item[(b)] \(\sum_{j=1}^4 A_{ij} = I_s\) for each row \(i\);
\item[(c)] \(\sum_{i=1}^4 A_{ij} = I_s\) for each column \(j\).
\end{enumerate}
Conjugation by \(\widetilde P_{C_4}^k\) just permutes rows and columns (and hence the
blocks) of \(A\), so each \(\widetilde P_{C_4}^{\,k} A \widetilde P_{C_4}^{-k}\) is again a QMS:
the blocks remain positive semidefinite and the row/column sums stay equal to \(I_s\).

\(\Phi(A)\) is a convex combination of these matrices, so:
\begin{itemize}
\item each block \(\Phi(A)_{ij}\) is a convex combination of positive semidefinite blocks, hence \(\Phi(A)_{ij}\succeq 0\);
\item the row-sums and column-sums of \(\Phi(A)\) are averages of \(I_s\), hence still equal to \(I_s\).
\end{itemize}
Therefore, since \(\Phi(A)\) satisfies (a), (b), (c), it is a quantum magic square and this proves (i).
\end{proof}

\subsection{Quantum Magic Squares as Free Spectrahedra}\label{sec:free-spec}
In this section, we study quantum magic squares
as free spectrahedra, following
\cite{DDN20}.
\begin{theorem}\label{thm:Mn-free-spect}
The set \(\cMn\) of quantum magic squares is affinely isomorphic to a compact free spectrahedron.
More precisely, after an affine translation, it can be described as
\[
\widehat{\mathcal M}^{(n)}_s
=
\Bigl\{\, Y = (Y_{ij})_{i,j=1}^{n} \in \operatorname{Mat}_n(\Her{s}) \,\Big|\,
I_{n^2}\otimes I_s
+ \sum_{i=1}^{n}\sum_{j=1}^{n} A_{ij}\otimes Y_{ij}
\succeq 0
\Bigr\},
\]
where the matrices \(A_{ij}\) arise from the affine parametrization of the magic relations.
\end{theorem}

A closely related point of view is given by \cite{BNS23}, where quantum magic squares are identified with the maximal matrix convex set associated to the (centered) Birkhoff polytope.
More precisely, in their notation \(\mathcal M^{(n)}_s = (B_N)_{\max}(s)\) via an explicit description in terms of linear inequalities.
Our approach here provides an explicit (after an affine translation) monic linear pencil adapted to magic relations, which we will later use to enforce commutation with the adjacency matrix of a graph.

In what follows, we give an explicit monic linear pencil for the quantum magic square.

\paragraph{Affine system for QMS.}

Fix \(n\ge 2\) and \(s\ge 1\), and consider block matrices \(X=(X_{kl})_{1\le k,l\le n}\) with \(X_{kl}\in\Her{s}\) satisfying the \emph{magic relations}, that is,
\[
\sum_{l=1}^n X_{kl}=I_s\quad(1\le k\le n),\qquad
\sum_{k=1}^n X_{kl}=I_s\quad(1\le l\le n),\qquad
X_{kl}\succeq 0\ \ \forall k,l.
\]
These are affine equations in the blocks \(X_{kl}\).
By solving the affine system, we can express the last row and the last column in terms of the
\((n-1)^2\) independent blocks \(\{X_{ij}\}_{1\le i,j\le n-1}\):
\begin{align*}
X_{i n} &= I_s - \sum_{j=1}^{n-1} X_{ij}\quad(1\le i\le n-1),\\
X_{n j} &= I_s - \sum_{i=1}^{n-1} X_{ij}\quad(1\le j\le n-1),\\
X_{n n} &= -(n-2)\,I_s + \sum_{i=1}^{n-1}\sum_{j=1}^{n-1} X_{ij}.
\end{align*}

Let the \((n-1)^2\) blocks \(\{X_{ij}\}_{1\le i,j\le n-1}\) be the independent variables; then each entry \(X_{kl}\) can be written as an affine combination:
\begin{equation}\label{eq:qms-affine-expansion}
X_{kl}\;=\;\alpha_{kl}\,I_s\;+\;\sum_{i=1}^{n-1}\sum_{j=1}^{n-1} c_{kl}^{\,ij}\,X_{ij},
\qquad 1\le k,l\le n,
\end{equation}
with coefficients \(\alpha_{kl}\in\{0,1,2-n\}\) and \(c_{kl}^{\,ij}\in\{-1,0,1\}\) determined by the relations above.
Collecting all blocks, this defines the affine parametrization
\[
\Psi_{\mathrm{magic}} : (\Her{s})^{(n-1)^2} \to \cMn_s,\qquad
\Psi_{\mathrm{magic}}(X_{11},X_{12},\dots,X_{n-1,n-1}) = (X_{kl})_{k,l=1}^n,
\]
whose image is exactly \(\cMn_s\).

\paragraph{Linear pencil collecting the positivity constraints.}

The inequalities \(X_{kl}\succeq 0\) can be put into a block-diagonal matrix, yielding the linear pencil
\[
L(X)\;=\;\operatorname{diag}\bigl(X_{11},X_{12},\dots,X_{nn}\bigr)\ \in\ \Mat{\R}{n^2}\otimes\Her{s}.
\]
Replacing \eqref{eq:qms-affine-expansion} gives
\begin{equation}\label{eq:qms-pencil}
L(X)\;=\;A_0\otimes I_s\;+\;\sum_{i=1}^{n-1}\sum_{j=1}^{n-1} A_{ij}\otimes X_{ij},
\end{equation}
where \(A_0, A_{ij}\in \Mat{\R}{n^2}\) are diagonal matrices with entries
\[
(A_0)_{(k,l),(k,l)}=\alpha_{kl},\qquad
(A_{ij})_{(k,l),(k,l)}=c_{kl}^{\,ij}.
\]
Thus \(X\in \cMn_s\) if and only if \(L(X)\succeq 0\).

\begin{remark}
Each block of the block-diagonal pencil \(L(X)\) can be viewed as
\[
\begin{pmatrix}
0 & & \\
& X_{kl} & \\
& & 0
\end{pmatrix} 
= I_{kl} \otimes X_{kl},
\]
where \(I_{kl}\) is the matrix with a \(1\) at the \((k,l)\)-th diagonal entry and \(0\) elsewhere. 
This shows explicitly how the diagonal entries of \(A_0\) and \(A_{ij}\) correspond to the coefficients \(\alpha_{kl}\) and \(c_{kl}^{ij}\) in \eqref{eq:qms-affine-expansion}.
\end{remark}

\paragraph{Monic linear form.}

Setting
\[
X_{ij}\;=\;\frac{1}{n}\,\bigl(I_s+Y_{ij}\bigr)\qquad(1\le i,j\le n-1).
\]
and substituting it into \eqref{eq:qms-pencil} gives us
\[
L(Y)\;=\;\Bigl(A_0+\frac{1}{n}\sum_{i,j}A_{ij}\Bigr)\otimes I_s
\;+\;\frac{1}{n}\sum_{i,j}A_{ij}\otimes Y_{ij}.
\]
Now we show that
\[
A_0+\frac{1}{n}\sum_{i,j}A_{ij}\;=\;\frac{1}{n}\,I_{n^2}.
\]
Indeed, the block matrix with all entries equal to \(\frac{1}{n}I_s\) satisfies \eqref{eq:qms-affine-expansion} for every \((k,l)\), and hence it is a quantum magic square. Choosing all independent variables equal to \(\frac{1}{n}I_s\) yields
\[
\frac{1}{n}I_s
\;=\;
\alpha_{kl}\,I_s+\frac{1}{n}\sum_{i,j}c_{kl}^{\,ij}\,I_s
\quad\Longrightarrow\quad
\alpha_{kl}+\frac{1}{n}\sum_{i,j}c_{kl}^{\,ij}\;=\;\frac{1}{n}.
\]
Thus each diagonal entry of \(A_0+\frac{1}{n}\sum_{i,j}A_{ij}\) equals \(\tfrac{1}{n}\), and we obtain
\[
L(Y)\;=\;\frac{1}{n}\,\Bigl(I_{n^2}\otimes I_s\;+\;\sum_{i,j}A_{ij}\otimes Y_{ij}\Bigr).
\]
Multiplying by \(n>0\) yields the \emph{monic} linear pencil
\[
\widehat L(Y)\;=\;I_{n^2}\otimes I_s\;+\;\sum_{i=1}^{n-1}\sum_{j=1}^{n-1} A_{ij}\otimes Y_{ij}\ \succeq\ 0.
\]
We can now prove Theorem~\ref{thm:Mn-free-spect}, since the explicit monic pencil constructed above characterizes the entire set \(\cM^{(n)}_s\).

\begin{proof}

The set \(\mathcal{M}^{(n)}\) is compact: it is closed, since it is defined by affine linear equalities and semidefinite
constraints, and bounded, because all blocks \(X_{ij}\succeq 0\) satisfy
\(X_{ij}\preceq I_s\) as a consequence of the magic relations
\eqref{eq:magic}.
Hence \(\mathcal{M}^{(n)}_s\) is compact for each level \(s\),
and so is the set \(\mathcal{M}^{(n)} = \bigcup_s \mathcal{M}^{(n)}_s\).
\end{proof} 

\begin{corollary}[cf.~\cite{DDN20}]
The set of quantum magic squares \(\mathcal{M}^{(n)}\) coincides with the matrix convex hull of its Arveson extreme points. 
\end{corollary}

This follows directly from the fact that \(\mathcal{M}^{(n)}\) is a compact free spectrahedron.
Indeed, by Theorem~\ref{thm:EvertHelton}, 
the matrix convex hull of the Arveson extreme points of a compact 
free spectrahedron coincides with the entire set.
\begin{remark}[Explicit coefficients]
An alternative proof follows by computing directly that \(A_0+\frac{1}{n}\sum_{i,j}A_{ij}=\tfrac{1}{n}I_{n^2}\), since the coefficients in \eqref{eq:qms-affine-expansion} are explicit:
\begin{itemize}
\item \(1\le k,l\le n-1\): \(\alpha_{kl}=0,\ c_{kl}^{\,ij}=\delta_{k,i}\delta_{l,j}\).
\item \(1\le i\le n-1\): \(X_{i n}=I_s-\sum_{j=1}^{n-1}X_{ij}\Rightarrow\alpha_{in}=1,\ c_{in}^{\,ij}=-\delta_{i,i}\).
\item \(1\le j\le n-1\): \(X_{n j}=I_s-\sum_{i=1}^{n-1}X_{ij}\Rightarrow\alpha_{nj}=1,\ c_{nj}^{\,ij}=-\delta_{j,j}\).
\item \(X_{nn}=-(n-2)I_s+\sum_{i,j=1}^{n-1}X_{ij}\Rightarrow\alpha_{nn}=2-n,\ c_{nn}^{\,ij}=1\).
\end{itemize}
With these values, we can verify that \(\alpha_{kl}+\frac{1}{n}\sum_{i,j}c_{kl}^{\,ij}=\frac{1}{n}\) for all \((k,l)\).
\end{remark}
\medskip
\noindent
In summary, the set of quantum magic squares admits a concrete description as a free spectrahedron, defined by the monic linear pencil \(\widehat L(Y)\). 
This description highlights the geometry of \(\cM^{(n)}\) as a compact matrix convex set: it is completely determined by finitely many linear matrix inequalities.
In the next section, we extend this framework to the set of graph quantum magic squares.
\subsection{Graph Quantum Magic Squares as Free Spectrahedra}\label{sec:GQMS_construction}

We now extend the previous construction to the case of graph quantum magic squares. 
The key idea is to encode the commutation relations with the adjacency matrix of the graph into the affine system and the resulting linear pencil.

For \(k\)-regular graphs, this set admits an explicit (monic) linear matrix inequality representation.
\begin{theorem}\label{thm:kregular_GQMS_freespec}
Let \(\Gamma\) be a \(k\)-regular graph and let \(s \ge 1\).
There exists an affine translation
\[
\widehat{\mathcal M}^{(\Gamma)}_s
=
\mathcal M^{(\Gamma)}_s - X^\circ
\]
such that the translated set admits a monic LMI representation of the form
\[
\widehat{\mathcal M}^{(\Gamma)}_s
=
\bigl\{\, Y \in \operatorname{Mat}_n(\Her{s}) \mid
I_{n^2}\otimes I_s
+ \sum_{(i',j')\in \mathcal{I}_{\mathrm{ind}}}
B_{i'j'} \otimes Y_{i'j'}
\succeq 0
\,\bigr\},
\]
where \(\mathcal{I}_{\mathrm{ind}}\) denotes the set of indices corresponding to the
independent variables in the affine parametrization.

In particular, the graded set
\[
\mathcal M^{(\Gamma)} = \bigcup_{s \ge 1} \mathcal M^{(\Gamma)}_s
\]
is affinely isomorphic to a compact free spectrahedron.
Hence \(\mathcal M^{(\Gamma)}\) is a compact matrix convex set.
\end{theorem}

In the next subsections, we construct this pencil explicitly, first for the simple case of complete graphs, and then for general \(k\)-regular graphs.

\subsubsection{Complete graphs}

For the complete graph \(K_n\), the adjacency matrix is
\[
A_{K_n} = J - I_n,
\]
where \(J\) denotes the \(n\times n\) all-ones matrix.
For any quantum magic square \(X=(X_{ij})\), the commutation condition
\[
(A_{K_n}\otimes I_s) X = X (A_{K_n}\otimes I_s)
\]
is automatically satisfied. Indeed, the magic relations imply that both row and column sum to the identity, hence
\[
(XJ)_{ik} = \sum_{j=1}^n X_{ij} = I_s, \qquad
(JX)_{ik} = \sum_{i=1}^n X_{ij} = I_s.
\]
Thus \(X(J\otimes I_s) = (J\otimes I_s)X\), and since \(I_n\) commutes with everything, we obtain
\[
(A_{K_n}\otimes I_s)X = ((J-I_n)\otimes I_s)X = X((J-I_n)\otimes I_s) = X (A_{K_n}\otimes I_s).
\]
Therefore, for \(K_n\) the commutation condition does not add any further restriction. We, hence, have
\[
\Ms[K_n] = \cMn_s,
\]
and the LMI coincides with the one obtained above.
\subsubsection{Affine parametrization under commutation for \(k\)-regular graphs}
Now, let \(\Gamma\) be a \(k\)-regular graph on \(n\) vertices with adjacency matrix \(A_\Gamma\). 
We consider block matrices \(X = (X_{ij})_{i,j=1}^n\) with \(X_{ij} \in \Her{s}\) satisfying both commutation with \(A_\Gamma\) and the magic relations.
The commutation constraints form a homogeneous linear system in the blocks \(\{X_{ij}\}\); let us denote the dimension of its solution space by \(d_\Gamma\).

Before deriving the parametrization, we analyze how the magic constraints interact with the commutation for \(k\)-regular graphs. The following proposition shows that the commutation condition forces row and column sums to be constant on each connected component.

\begin{proposition}[Componentwise row/column sums]\label{prop:regular_sums}
Let \(\Gamma\) be a \(k\)-regular graph with \(N\) connected components
\(\Gamma=\bigcup_{t=1}^N \Gamma_t\) and let \(A_\Gamma\) be its adjacency matrix.

Let \(X=\left(X_{ij}\right)_{i,j=1}^n\in \mathrm{Mat}_n(\mathrm{Her}(s))\) be such that
\[
X(A_\Gamma\otimes I_s)=(A_\Gamma\otimes I_s)X.
\]
Then there exist Hermitian matrices \(\Lambda_R^{(t)},\Lambda_C^{(t)}\in \mathrm{Her}(s)\) such that, for every \(t=1,\dots,N\),
\[
\sum_{j=1}^n X_{ij}=\Lambda_R^{(t)} \quad (\forall i\in V_t),
\qquad
\sum_{i=1}^n X_{ij}=\Lambda_C^{(t)} \quad (\forall j\in V_t),
\]
where \(V_t\) is the set of vertices of the component \(\Gamma_t\).
\end{proposition}
\bigskip
Before proceeding with the proof we recall a property of the eigenspaces of a \(k\)-regular graph.
\begin{observation}
If \(\Gamma\) is \(k\)-regular with \(n\) vertices, then \(A_{\Gamma}\mathbf 1 = k\mathbf 1\), where \(\mathbf 1=(1,\dots,1)^T \in \C^n\).
In particular, if \(\Gamma = \bigsqcup_{t=1}^N \Gamma_t\) is a \(k\)-regular graph with \(N\) connected components, it also holds that
\[
A_\Gamma \mathbf 1_t = k\,\mathbf 1_t \qquad \forall t=1,\dots,N,
\]
where \(\mathbf 1_t\) is the \(n\)-dimensional vector defined by
\[
(\mathbf 1_t)_i=
\begin{cases}
1, & i\in V(\Gamma_t),\\
0, & i\notin V(\Gamma_t),
\end{cases}
\qquad i=1,\dots,n.
\]
Furthermore, since each \(\Gamma_t\) is connected and \(k\)-regular, the eigenspace relative to the eigenvalue \(k\) of \(A_\Gamma\) is \(\ker(A_\Gamma-kI)=\mathrm{span}\{\mathbf 1_t \mid t=1,\dots,N\}\).
So the eigenspace related to the eigenvalue \(k\) of \(A_\Gamma\otimes I_s\) is
\[
\ker(A_\Gamma\otimes I_s-kI)
=
\mathrm{span}\{\mathbf 1_t \otimes e_p \mid t=1,\dots,N, \; p=1,\dots,s\},
\]
where $\{e_1, \dots, e_s\}$ is the canonical basis of $\mathbb{C}^s$.
\end{observation}

\begin{proof}

Let \(V=\{1,\dots,n\}\) be the set of vertices of \(\Gamma=\bigsqcup_{t=1}^N\Gamma_t\),
and \(V_t\) the set of vertices of \(\Gamma_t\).

Now fix an index \(t_0\in\{1,\dots,N\}\) and a vector \(v \in\mathbb C^s\).
From the commutation \(X(A_\Gamma\otimes I_s)=(A_\Gamma\otimes I_s)X\) it follows that
\[
(A_\Gamma\otimes I_s) \left[ X(\mathbf 1_{t_0}\otimes v) \right]
=
X(A_\Gamma\otimes I_s)(\mathbf 1_{t_0}\otimes v)
=
X(k\,\mathbf 1_{t_0}\otimes v)
=
k \left[ X(\mathbf 1_{t_0}\otimes v) \right],
\]
so \(X(\mathbf 1_{t_0}\otimes v)\) is an eigenvector of \(A_\Gamma\otimes I_s\) with eigenvalue \(k\). Therefore there exist coefficients \(c_{tp}(v) \in \mathbb{C}\) such that
\[
X(\mathbf 1_{t_0}\otimes v) = \sum_{t=1}^N \sum_{p=1}^s c_{tp}(v) (\mathbf 1_t \otimes e_p).
\]

The \(i\)-th component of \(X(\mathbf 1_{t_0}\otimes v)\) is
\begin{equation}\label{eq:i_comp_X_v}
    \left[ X(\mathbf 1_{t_0}\otimes v) \right]_i
=
\sum_{j=1}^n X_{ij}(\mathbf 1_{t_0})_j\, v
=
\sum_{j\in V_{t_0}} X_{ij} v.
\end{equation}
On the other hand, if \(i\in V_\ell\) then \((\mathbf 1_\ell)_i=1\) and \((\mathbf 1_t)_i=0\) for \(t\neq \ell\), so from the previous decomposition we obtain
\begin{equation}\label{eq:i_th_block}
\left[ X(\mathbf 1_{t_0}\otimes v) \right]_i = \sum_{p=1}^s c_{\ell p}(v) e_p \qquad (\forall i\in V_\ell).
\end{equation}

So, comparing \eqref{eq:i_comp_X_v} and \eqref{eq:i_th_block}, for each \(\ell\) and for each \(i\in V_\ell\) we have
\[
\sum_{j\in V_{t_0}} X_{ij} v = \sum_{p=1}^s c_{\ell p}(v) e_p.
\]
Since the right-hand side depends only on the component index \(\ell\) (and the vector \(v\)), but not on the specific vertex \(i\), it follows that the matrix sum \(\sum_{j\in V_{t_0}} X_{ij}\) is the same for all \(i \in V_\ell\). 

In particular, by summing over all \(t_0=1,\dots,N\), we obtain that for every \(i \in V_\ell\):
\[
\sum_{j=1}^n X_{ij} = \sum_{t_0=1}^N \left( \sum_{j\in V_{t_0}} X_{ij} \right) =: \Lambda^{(\ell)}_{R},
\]
where \(\Lambda^{(\ell)}_{R}\in \mathrm{Mat}_s(\mathbb C)\) depends only on the connected component \(V_\ell\). 
Since \(X_{ij} \in \mathrm{Her}(s)\), the sum \(\Lambda^{(\ell)}_{R}\) is also Hermitian.
For the columns, we follow the same reasoning with \(X^T\) instead of \(X\):
since \(A_\Gamma^T=A_\Gamma\), the commutation also holds,
\[
X^T(A_\Gamma\otimes I_s)=(A_\Gamma\otimes I_s)X^T.
\]
Applying the above to \(X^T\) we obtain that 
\[
\sum_{i=1}^n X_{ij}
=:\Lambda^{(\ell)}_{C}\in \mathrm{Mat}_s(\mathbb C),
\]
the column sums of \(X\) are constant on each connected component \(l\), yielding the same matrix \(\Lambda^{(\ell)}_{C}\).
\end{proof}

In general, magic relations require that all row and column sums be equal to \(I_s\).
However, thanks to Proposition~\ref{prop:regular_sums}, for a matrix in the commutant of \(A_\Gamma\), it is sufficient to impose this condition once per connected component.
In particular, imposing \(\Lambda_R^{(t)} = I_s\) (and similarly for the columns) fixes the sum for all rows of that component.
Consequently, imposing magic relations removes exactly \(N\) degrees of freedom (where \(N\) is the number of connected components) from the \(d_\Gamma\) independent parameters provided by imposing the commutation constraints.

Now we perform the explicit parameterization.
Let
\[
\mathcal{I} \subseteq \{1,\dots,n\}^2, \qquad |\mathcal{I}| = d_\Gamma,
\]
be a set of independent position indices for the commutant. In practice, we choose the variables
\(X_{i'j'}\) with \((i',j')\in \mathcal{I}\) from the first and second rows
until exactly \(d_\Gamma\) independent parameters are collected.

Solving the commutation constraints yields a linear parametrization
\[
\begin{alignedat}{2}
  &\Psi_{\mathrm{com}} : (\Her{s})^{\cI} &&\longrightarrow 
  \bigl\{\, X \in (\Her{s})^{n\times n} \mid A_\Gamma X = X A_\Gamma \,\bigr\}, \\[4pt]
  &\text{given by}\quad 
  &&\bigl(\Psi_{\mathrm{com}}((X_{i'j'})_{(i',j')\in I})\bigr)_{ij}
    \;=\; \sum_{(i',j')\in \cI} r^{\,ij}_{i'j'}\, X_{i'j'}.
\end{alignedat}
\]

Magic constraints allow us to express \(N\) of these variables as functions of the others.
Let \(\mathcal{K} \subset \mathcal{I}\) be the set of \(N\) dependent indices and \(\mathcal{I}_{\mathrm{ind}} = \mathcal{I} \setminus \mathcal{K}\) be the set of the remaining independent indices.

After the normalization
\[
X_{i'j'} \;=\; \frac{1}{n}\bigl(I_s + Y_{i'j'}\bigr),
\qquad (i',j')\in \mathcal{I}_{\mathrm{ind}},
\]
we obtain an affine parametrization of the form
\[
X_{kl} \;=\; \beta_{kl}\, I_s
  \;+\; \sum_{(i',j')\in \mathcal{I}_{\mathrm{ind}}} d_{kl}^{(i'j')}\, Y_{i'j'}.
\]
Collecting the inequalities \(X_{kl}\succeq 0\) gives the block-diagonal pencil
\[
L(Y) \;=\; B_0 \otimes I_s
  \;+\; \sum_{(i',j')\in \mathcal{I}_{\mathrm{ind}}} B_{i'j'} \otimes Y_{i'j'}.
\]
Finally, we observe that the constant matrix \(\left(X^\circ_{ij}\right)_{ij}=\left(\tfrac{1}{n}I_s\right)_{ij}\) satisfies both commutation and magic constraints. 
Indeed, for any \(k\)-regular graph, the all-ones matrix \(J\) commutes with \(A_\Gamma\) (since \(A_\Gamma J = J A_\Gamma = kJ\)), and thus the scalar matrix \(\tfrac{1}{n}I_s\) is a valid solution.
This implies that the constant term in our parametrized pencil corresponds to this point, and after rescaling, we obtain the monic pencil of Theorem~\ref{thm:kregular_GQMS_freespec}:
\[
\widehat L(Y)
 \;=\; I_{n^2} \otimes I_s
   \;+\; \sum_{(i',j')\in \mathcal{I}_{\mathrm{ind}}} B_{i'j'} \otimes Y_{i'j'}
 \;\succeq\; 0.
\]
\subsubsection{Composition of affine maps and explicit coefficients}

In the previous section, we defined the matrices \(B_{i'j'}\) implicitly via the affine system. 
Here we provide their explicit construction by composing the parametrization of the general QMS with the parametrization of the commutant.

Recall that in the general QMS case, each block admits the expansion
\[
X_{kl} \;=\; \alpha_{kl} I_s
  \;+\; \sum_{i,j=1}^{n-1} c_{kl}^{ij}\, X_{ij},
\]
where the coefficients \(\alpha_{kl}\) and \(c_{kl}^{ij}\) are explicitly determined by the magic relations (taking values respectively in \(\{ 0, 1, 2-n\}\) and \(\{-1, 0, 1\}\)).

In the GQMS case, we restrict the variables \(X_{ij}\) to the subspace of solutions to the commutation constraints. 
Using the linear map \(\Psi_{\mathrm{com}}\) defined above, we can express every block \(X_{ij}\) in terms of the independent variables indexed by \(\mathcal{I}\):
\[
X_{ij} = \bigl(\Psi_{\mathrm{com}}((X_{i'j'})_{(i',j')\in \mathcal{I}})\bigr)_{ij}
 \;=\; \sum_{(i',j')\in \mathcal{I}_{\mathrm{ind}}} r^{\,ij}_{i'j'}\, X_{i'j'},
\]
where the coefficients \(r^{\,ij}_{i'j'}\) come from solving the homogeneous system \((A_\Gamma \otimes I_s) X = X  ( A_\Gamma \otimes I_s)\).

After composing this with the affine parametrization for QMS, we get
\[
\begin{aligned}
X_{kl}
  &= \alpha_{kl} I_s
   + \sum_{i,j=1}^{n-1}
       c_{kl}^{ij}\,
       \bigl(\Psi_{\mathrm{com}}((X_{i'j'})_{(i',j')\in \mathcal{I}})\bigr)_{ij} \\
  &= \alpha_{kl} I_s
   + \sum_{i,j=1}^{n-1} c_{kl}^{ij}
       \sum_{(i',j')\in \mathcal{I}_{\mathrm{ind}}} r^{\,ij}_{i'j'}\, X_{i'j'} \\
  &= \alpha_{kl} I_s
   + \sum_{(i',j')\in \mathcal{I}_{\mathrm{ind}}}
       \Bigl( \sum_{i,j=1}^{n-1} c_{kl}^{ij}\, r^{\,ij}_{i'j'} \Bigr)\, X_{i'j'}.
\end{aligned}
\]
Hence, defining 
\[
d_{kl}^{(i'j')} := \sum_{i,j=1}^{n-1} c_{kl}^{ij}\, r^{\,ij}_{i'j'},
\]
we get that the explicit form of  \(B_{i'j'}\) in the LMI of Theorem~\ref{thm:kregular_GQMS_freespec} is given by diagonal matrices whose entries are precisely the coefficients \(d_{kl}^{(i'j')}\).

Now, if we perform the shift
\[
X_{i'j'} = \frac{1}{n}\bigl(I_s + Y_{i'j'}\bigr),
\qquad (i',j')\in \mathcal{I}_{\mathrm{ind}},
\]
we get exactly the monic affine parametrization
\begin{equation}\label{eq:GQMS_affine_parametrization}
X_{kl} \;=\; \frac{1}{n} I_s
  \;+\; \sum_{(i',j')\in \mathcal{I}_{\mathrm{ind}}} d_{kl}^{(i'j')} \, Y_{i'j'}.
\end{equation}
This confirms that the number of independent Hermitian parameters corresponds exactly to \(d_\Gamma - N\) when \(\Gamma\) has \(N\) connected components. An explicit calculation of this dimension for the cycle graphs \(C_n\) is provided in Appendix~\ref{sec:commutant-dimension}.

\subsection{GQMS and Arveson extreme points}
Theorem~\ref{thm:kregular_GQMS_freespec} establishes that
\(\mathcal{M}^{(\Gamma)}\) is affinely isomorphic to a compact free spectrahedron. We may now apply Theorem~\ref{thm:EvertHelton} to conclude that
\[
\mathcal{M}^{(\Gamma)}
\;=\;
\operatorname{mconv}\bigl(
\mathrm{ArvesonExt}(\mathcal{M}^{(\Gamma)})
\bigr).
\]

In the non-graph case, \cite{DDN20} showed that every quantum
permutation matrix is an Arveson extreme point of the free spectrahedron
\(\mathcal{M}^{(n)}\).
Since graph quantum permutation matrices form a subclass of
quantum permutation matrices subject only to additional linear constraints,
the same argument used in \cite{DDN20} also applies
in the graph setting and we have the following corollary.

\begin{corollary}\label{cor:graph-arveson}
Let \(\Gamma\) be a graph on \(n\) vertices.
Every graph quantum permutation matrix in \(\mathcal{P}^{(\Gamma)}\)
is an Arveson extreme point of the matrix convex set
\(\mathcal{M}^{(\Gamma)}\) of graph quantum magic squares.
\end{corollary}
\begin{proof}
By definition we have
\[
\mathcal{M}^{(\Gamma)} \subseteq \mathcal{M}^{(n)}, 
\qquad 
\mathcal{P}^{(\Gamma)} \subseteq \mathcal{P}^{(n)}.
\]
Let \(U \in \mathcal{P}^{(\Gamma)}\). 
Then \(U \in \mathcal{P}^{(n)}\), so by 
\cite{DDN20}*{Corollary~20}
\(U\) is an Arveson extreme point of \(\mathcal{M}^{(n)}\).

Now let \(\Tilde{U}\) be any dilation of \(U\) inside \(\mathcal{M}^{(\Gamma)}\), 
in the sense of Definition~\ref{def:arveson}.
Since \(\mathcal{M}^{(\Gamma)} \subseteq \mathcal{M}^{(n)}\), 
\(\Tilde{U}\) is also a dilation of \(U\) inside \(\mathcal{M}^{(n)}\).
Because \(U\) is Arveson extreme in \(\mathcal{M}^{(n)}\), 
this dilation must be trivial, i.e.\ all blocks of \(\Tilde{U}\) split as
\[
\Tilde{u}_{ij} =
\begin{pmatrix}
u_{ij} & 0 \\
0      & \gamma_{ij}
\end{pmatrix}
\]
up to a simultaneous unitary conjugation.
Therefore \(U\) is Arveson extreme also in \(\mathcal{M}^{(\Gamma)}\).
\end{proof}

\subsubsection{Concluding Remarks on the Free-Spectrahedral Setting}

The set \(\mathcal{M}^{(\Gamma)}\) is  affinely isomorphic to a compact free spectrahedron,
defined by finitely many linear matrix inequalities with scalar coefficients.
By Theorem~\ref{thm:EvertHelton},
its matrix convex hull is generated by its Arveson extreme points.

Corollary~\ref{cor:graph-arveson} shows that every graph quantum permutation
matrix belongs to the set \(\operatorname{ArvesonExt}(\mathcal{M}^{(\Gamma)})\).
Hence, in analogy with the non-graph case
\cite{DDN20}*{Corollary~20}, the set
\(\mathcal{M}^{(\Gamma)}\) is strictly larger than that the set
\(\mathcal{P}^{(\Gamma)}\) whenever the Birkhoff--von Neumann property fails
for~\(\Gamma\).

\begin{observation}[\(\mathcal{M}^{(n)}\) and \(\mathcal{M}^{(\Gamma)}\) over a unital \(C^*\)-algebra]
If we take any unital \(C^*\)-algebra \(\mathcal{A}\), all previous results
hold upon replacing \(\Her{s}\) with \(\mathcal{A}_{\mathrm{sa}}\)
and \(I_s\) with \(I_{\mathcal{A}}\).
\end{observation}

\section{Future Directions}

These results suggest some avenues for future work.

\subsection{Beyond \(C_4\): the (Quantum) Birkhoff--von Neumann problem for graphs}

Having found an example of the failure of the Birkhoff--von Neumann theorem for the graph \(C_4\), suggests investigating whether
the Birkhoff--von Neumann theorem fails in the quantum setting for other graphs.
A natural class of candidates is vertex-transitive or symmetric graphs,
such as the cycle \(C_5\), the Petersen graph, and, more generally, any \(k\)-regular graph
with a large automorphism group.

This leads to the conjecture:
\[
\operatorname{mconv}(\mathcal P^{(\Gamma)}) 
\subsetneq 
\mathcal M^{(\Gamma)} \qquad \text{for various graphs }\Gamma.
\]

Both analytic and numerical tests (e.g.\ via symmetry-reduced SDP
formulations) would provide evidence toward a general theory.

\subsection{Arveson extreme points of \(\mathcal M^{(\Gamma)}\)}

For the QMS case, \cites{DDN20, EH19} show that
\[
\mathcal M^{(n)} 
= \operatorname{mconv}\big(\operatorname{ArvesonExt}(\mathcal M^{(n)})\big)
\quad\text{and}\quad
\mathcal P^{(n)} \subseteq \operatorname{ArvesonExt}(\mathcal M^{(n)}).
\]

For the GQMS case, Corollary~\ref{cor:graph-arveson} establishes the
analogous statement:
\[
\mathcal P^{(\Gamma)} \subseteq \operatorname{ArvesonExt}\big(\mathcal M^{(\Gamma)}\big),
\]
i.e.\ every graph quantum permutation matrix is an Arveson extreme point
of the graph quantum magic square spectrahedron.

Then we can ask whether we can characterise the full family of Arveson extreme points.

To identify the ``vertices'' of \(\mathcal M^{(\Gamma)}\), we need a better understanding of its extremal structure, and this also helps us compare the graph and non-graph settings in \cite{DDN20}.

\subsection{Connections with quantum information theory}

Quantum magic squares can be formulated in terms of positive operator-valued measures (POVMs) and thus connect to quantum information theory, for example via non-local games.
Since quantum permutation matrices correspond to projection-valued measurement (PVM), it is natural to ask whether the failure of a Birkhoff-von Neumann theorem type for quantum magic squares can be interpreted as a separation between POVM and PVM.

Hence, a natural direction for future work is to investigate this perspective for Graph Quantum Magic Squares (GQMS) and to identify graphs \(\Gamma\) for which \(\mathcal M^{(\Gamma)}\) exhibits genuinely different behaviors in the PVM vs POVM models. Clarifying these points would further connect GQMS to quantum information theory.

\appendix
\section{Appendix}

\subsection{Dual SDP and Explicit Certificate}\label{sec:dual-certificate}

In this appendix we briefly explain the semidefinite program used in the proof of
Theorem~\ref{thm:C4-counterexample}.

Recall that \(M_B = \varphi(B) + \psi(B)\) and that
\[
\mathcal{S}
 := \bigl(\cZ_e^{(4)} \otimes \cZ_e^{(4)}
          \otimes \Mat{\C}{2}\bigr)_{\mathrm{her}}
\]
is the Hermitian subspace of
Proposition~\ref{prop:DDN-separation}.
If \(B \in \mathrm{mconv}(\cP^{(4)})_2\), then there exists
\(X \in \mathcal{S}\) such that \(M_B + X \succeq 0\).

Writing \(\{Z_{a,b,j}\}\) for a fixed Hermitian basis of \(\mathcal{S}\),
every \(X \in \mathcal{S}\) can be written as
\[
X = \sum_{a,b,j} \xi_{a,b,j} Z_{a,b,j}.
\]
Thus, the condition \(M_B+X\succeq 0\) can be reformulated
in terms of coefficients \(\xi_{a,b,j}\) such that
\[
M_B + \sum_{a,b,j} \xi_{a,b,j} Z_{a,b,j} \succeq 0.
\]

The corresponding dual problem look for a positive semidefinite
matrix \(Y \succeq 0\) such that
\[
\Tr\!\bigl(Y Z_{a,b,j}\bigr) = 0
\quad\text{for all }a,b,j,
\]
and attempts to minimize \(\Tr(Y M_B)\).
If the optimal value is negative, then every admissible \(X\) satisfies
\[
\Tr\!\bigl(Y(M_B+X)\bigr)
 = \Tr(Y M_B) + \Tr(YX)
 = \Tr(Y M_B) < 0,
\]
which excludes the possibility that \(M_B + X\) is positive semidefinite.
Hence \(B \notin \mathrm{mconv}(\cP^{(4)})_2\).

We found \(Y\) by using the
CVXPY library in Python.  

This provides the construction of the matrix \(Y\) used in
Theorem~\ref{thm:C4-counterexample}.

\subsubsection{Explicit basis for \(\mathcal{Z}^{(4)}_e\)}
\label{sec:explicit-basis-Z4}

For completeness, we give an explicit Hermitian basis of the subspace
\[
\mathcal{Z}^{(4)}_e=\{Z\in \Mat{\C}{4}\mid\ \mathrm{diag}(Z)=0,\ Z\mathbf{1}=0,\ Z^{*}\mathbf{1}=0\},
\quad \mathbf{1}=(1,1,1,1)^{\mathsf{t}},
\]
which has dimension \(5\).
We chose the following matrices:
\[
z_1=\begin{pmatrix}
0&1&0&-1\\[2pt] 1&0&-1&0\\[2pt] 0&-1&0&1\\[2pt] -1&0&1&0
\end{pmatrix},\quad
z_2=\begin{pmatrix}
0&i&0&-i\\[2pt] -i&0&0&i\\[2pt] 0&0&0&0\\[2pt] i&-i&0&0
\end{pmatrix},\quad
z_3=\begin{pmatrix}
0&0&1&-1\\[2pt] 0&0&-1&1\\[2pt] 1&-1&0&0\\[2pt] -1&1&0&0
\end{pmatrix},
\]
\[
z_4=\begin{pmatrix}
0&0&i&-i\\[2pt] 0&0&0&0\\[2pt] -i&0&0&i\\[2pt] i&0&-i&0
\end{pmatrix},\quad
z_5=\begin{pmatrix}
0&0&0&0\\[2pt] 0&0&i&-i\\[2pt] 0&-i&0&i\\[2pt] 0&i&-i&0
\end{pmatrix}.
\]
For \(\Mat{\C}{2}\) we use the standard Hermitian basis
\[
s_1=\begin{pmatrix}1&0\\0&0\end{pmatrix},\quad
s_2=\begin{pmatrix}0&1\\1&0\end{pmatrix},\quad
s_3=\begin{pmatrix}0&-i\\ i&0\end{pmatrix},\quad
s_4=\begin{pmatrix}0&0\\0&1\end{pmatrix}.
\]
Writing \(Y_{i,j,k}:=z_i\otimes z_j\otimes s_k\), we can describe an explicit Hermitian basis for
\[
\mathcal{S}=\big(\mathcal{Z}^{(4)}_e\otimes \mathcal{Z}^{(4)}_e\otimes \Mat{\C}{2}\big)_{\mathrm{her}}
\]
by setting
\[
C_{i,j,k}:=Y_{i,j,k}+Y_{i,j,k}^{*},\qquad
D_{i,j,k}:=i\big(Y_{i,j,k}-Y_{i,j,k}^{*}\big),
\]
so that \(\mathcal{S}\) is spanned by arbitrary real linear combinations of \(C_{i,j,k}\) and \(D_{i,j,k}\).
(Any other basis of \(\mathcal{Z}^{(4)}_e\) would work as well.)

\bigskip
\noindent
As an example, we now compute the dimension of the commutant of the adjacency matrix
for cyclic graphs \(C_n\).

\subsection{Explicit computation of the commutant dimension for cycle graphs}
\label{sec:commutant-dimension}

In Section~\ref{sec:GQMS_construction} (specifically Theorem~\ref{thm:kregular_GQMS_freespec}), we showed that the number of independent Hermitian parameters for a \(k\)-regular graph \(\Gamma\) is given by \(d_\Gamma - N\), where \(d_\Gamma = \dim(\operatorname{Comm}(A_\Gamma))\) and \(N\) is the number of connected components.

Here, we compute this dimension explicitly for the cycle graphs \(C_n\), which corresponds to the case of the counterexample discussed in Section~\ref{sec:counterexample}.

\subsubsection{Cycle graphs \(C_n\)}

Let \(d_n=\dim\{X\in \operatorname{Mat}_n(\C)\mid (A_{C_n}\otimes I_s)X=X(A_{C_n}\otimes I_s)\}\) for \(C_n\), and recall that \(C_n\) is connected and \(2\)-regular.

\begin{proposition}
Let \(A_{C_n}\in \operatorname{Mat}_n(\R)\) be the adjacency matrix of \(C_n\) and set
\[
\operatorname{Comm}(A_{C_n}) := \{ X \in \operatorname{Mat}_n(\C) \mid (A_{C_n}\otimes I_s)X = X(A_{C_n}\otimes I_s) \},
\qquad
d_n := \dim\operatorname{Comm}(A_{C_n}).
\]
Then
\[
d_n =
\begin{cases}
2n - 1, & \text{\(n\) odd},\\[1mm]
2n - 2, & \text{\(n\) even}.
\end{cases}
\]
\end{proposition}

\begin{proof}
Since \(A_{C_n}\) is real symmetric, it is diagonalizable; the dimension of its
commutant in \(\operatorname{Mat}_n(\C)\) is \(\sum_{\lambda} m_\lambda^2\), where \(m_\lambda\) is
the multiplicity of the eigenvalue \(\lambda\).

The eigenvalues of \(A_{C_n}\) are
\(\lambda_k = 2\cos\!\bigl(\tfrac{2\pi k}{n}\bigr)\), \(k=0,\dots,n-1\).

For \(n\) odd, \(\lambda_0=2\) has multiplicity \(1\). The other indices
appear in pairs \((k,n-k)\) with \(1\le k\le (n-1)/2\), each of them with multiplicity \(2\).
Thus
\[
d_n = 1^2 + \frac{n-1}{2}\cdot 2^2 = 2n-1.
\]
For \(n\) even, \(\lambda_0=2\) and \(\lambda_{n/2}=-2\) have multiplicity \(1\). The remaining indices
pair as \((k,n-k)\) with \(1\le k\le n/2-1\), each with multiplicity \(2\).
Hence
\[
d_n = 1^2 + 1^2 + \Bigl(\frac{n}{2}-1\Bigr)\cdot 2^2
= 2 + 2(n-2) = 2n-2.
\]
\end{proof}

\begin{corollary}
For a QMS commuting with \(A_{C_n}\), the number of independent
Hermitian block variables is \(d_n - 1\), i.e.,
\[
\#\text{ independent parameters } =
\begin{cases}
2n - 2, & \text{\(n\) odd},\\[1mm]
2n - 3, & \text{\(n\) even}.
\end{cases}
\]
\end{corollary}

\begin{remark}
The subtraction of \(1\) takes into account the single component (\(N=1\)) and the
resulting componentwise doubly stochastic constraint (row/column sums equal to \(I_s\)).
In the even case, the index pairing \((k,n-k)\) is already reflected in the
multiplicity pattern that produces \(d_n\).
\end{remark}

\section*{Acknowledgements}

I warmly thank Makoto Yamashita for his continuous guidance, for many helpful discussions,
and for his supervision throughout this project.
I am also grateful to  Inga  Valentiner-Branth, and Tim Netzer\ for their clarifications regarding
some technical aspects of their work, and for their kind replies to my questions.
I would also like to thank Tea Štrekelj for carefully reading an earlier version of the manuscript
and for her helpful comments, in particular on aspects related to matrix convexity
and free spectrahedra.

This research was funded by The Research Council of Norway [project 324944].
\section*{Data Availability}

The author declare that the data supporting findings of this study are available within the paper and its supplementary information files.

\end{document}